\documentclass[
aps,pra,
reprint,
a4paper,
longbibliography,
floatfix,
]{revtex4-1}
\usepackage[utf8]{inputenc}

\usepackage{amsmath,amssymb,amsthm}
\usepackage{mathtools}
\usepackage{braket}
\usepackage{hyperref}
\usepackage{xcolor}

\usepackage{array, makecell}%
\setcellgapes{4pt}

\hypersetup{citecolor=blue}
\hypersetup{colorlinks=true}
\hypersetup{linkcolor=blue}
\hypersetup{urlcolor=blue}

\theoremstyle{definition}
\newtheorem{theorem}{Theorem}

\newtheorem{lemma}[theorem]{Lemma}
\newtheorem{corollary}[theorem]{Corollary}
\DeclareMathOperator\COR{COR}
\DeclareMathOperator\SIM{SIM}
\DeclareMathOperator\USD{USD}
\DeclareMathOperator\tr{tr}
\DeclarePairedDelimiter\abs\lvert\rvert

\newcommand\proj[1]{\ket{#1}\!\bra{#1}}

\definecolor{darkred}{rgb}{.8 .1 .1}
\definecolor{darkgreen}{rgb}{.1 .8 .1}
\definecolor{darkyellow}{rgb}{.6 .6 .0}

\begin{document}

\title{Minimal scheme for certifying three-outcome qubit measurements in the prepare-and-measure scenario}

\author{Jonathan Steinberg}
\email{steinberg@physik.uni-siegen.de}
\affiliation{Naturwissenschaftlich–Technische Fakultät, Universität Siegen, 57068 Siegen, Germany}
\author{H. Chau Nguyen}
\affiliation{Naturwissenschaftlich–Technische Fakultät, Universität Siegen, 57068 Siegen, Germany}
\author{Matthias Kleinmann}
\affiliation{Naturwissenschaftlich–Technische Fakultät, Universität Siegen, 57068 Siegen, Germany}
\affiliation{Faculty of Physics, University of Duisburg--Essen, 47048 Duisburg, Germany}

\begin{abstract}
The number of outcomes is a defining property of a quantum measurement, in particular, if the measurement cannot be decomposed into simpler measurements with fewer outcomes. Importantly, the number of outcomes of a quantum measurement can be irreducibly higher than the dimension of the system. The certification of this property is possible in a semi-device-independent way either based on a Bell-like scenario or by utilizing the simpler prepare-and-measure scenario. Here we show that in the latter scenario the minimal scheme for certifying an irreducible three-outcome qubit measurement requires three state preparations and only two measurements and we provide experimentally feasible examples for this minimal certification scheme. We also discuss the dimension assumption characteristic to the semi-device-independent approach and to which extend it can be mitigated.
\end{abstract}

\maketitle

\section{Introduction}
The most general description of a quantum measurement is given by a positive operator valued measure (POVM) which is any collection of positive operators summing up to identity. The set of measurements described in this way contains instances which are neither projective nor can they be obtained by combining projective measurements. This class of genuinely nonprojective measurements can be utilized, for example, in quantum computing \cite{A21,A22}, quantum cryptography \cite{A25,A26}, randomness certification \cite{A24}, and quantum tomography \cite{A23}. But since genuinely nonprojective measurements cannot be combined from projective measurements, their experimental implementation is difficult and typically requires control over additional degrees of freedoms \cite{Peres:1993}. It is hence of interest to verify whether an experiment had successfully implemented a nonprojective measurement. Recently semi-device-independent certification schemes have come into the focus of theoretical investigation \cite{Mironowicz2019SW,Tavakoli2020bSW, Miklin2020SW} as well as experimental implementation \cite{A12, Tavakoli2020aSW,Smania2020}. The employed certification schemes can be divided into two classes, those based on Bell-like scenarios \cite{A12} and those using a prepare-and-measure scenario \cite{Mironowicz2019SW, Tavakoli2020bSW, Miklin2020SW, Tavakoli2020aSW}

In the former case, an entangled state is distributed to two spatially separated measurement stations and the correlations between the different measurements at each station can then certify the presence of a genuinely nonprojective measurement. In the latter case, the certification consists of several preparation procedures, possibly intermediate transformations, and subsequent measurements on the same system. Both scenarios are device-independent because only very rudimentary assumptions need to be made about the implementation details of the state preparation and measurement devices. However, because nonprojective measurements can always be implemented on a system with enlarged Hilbert space, it is common to both scenarios that they add an extra assumption, namely an upper limit on the dimension of the system, rendering the scenarios semi-device-independent.

It must be noted that the certification of a genuinely nonprojective measurements is usually performed in a slightly different context where the number of outcomes of a measurement plays a key role \cite{Mironowicz2019SW,Tavakoli2020bSW, Miklin2020SW, A12, Tavakoli2020aSW}. If a measurement cannot be implemented by using measurements with a lower number of outcomes, then the number of outcomes is irreducible. Since projective measurements cannot have more outcomes than the dimension of the system, it is hence sufficient to certify excess outcomes in order to certify that a measurement is also genuinely nonprojective. For the purposes of this paper, where we consider three-outcome measurements on a qubit, this distinction is not relevant, because for qubits, the number of outcomes of a measurement is irreducibly three if and only if the measurement is genuinely nonprojective.

In the prepare-and-measure scenario, so far, the certification schemes for three-outcome and four-outcome qubit measurements are based on linear inequalities satisfied for all correlations requiring less outcomes. The schemes use at least two additional measurement settings \cite{Tavakoli2020aSW} to achieve a certification. In this paper we do not restrict ourselves to linear inequalities and we show that the minimal scheme consists of three different state preparations and only one auxiliary measurement. Our analysis is complemented with examples which allow a simple certification of a three-outcome qubits measurement.

Our paper is organized as follows. In Section~\ref{PrepMeasure} we revisit the concept of prepare-and-measure scenarios and fix our notation and terminology. In Section~\ref{Certification} we define the operational setup in which a three-outcome qubit measurement can be certified. We prove the necessity of one auxiliary measurement and three preparation procedures and give corresponding examples. We proceed in Section~\ref{Spekkens} by discussing a possible mitigation of the assumption of an upper bound on the dimension before we conclude with a discussion in Section~\ref{conclusion}.

\section{Prepare-and-measure scenario} \label{PrepMeasure}
A prepare-and-measure scenario can be understood as a setup that is composed solely of a preparation device and a measurement device. An experimenter can choose among $s$ different preparation procedures labeled by $x \in \set{1,\dotsc,s}$ and $m$ different measurements labeled by $y \in \set{1,\dotsc,m}$. After choosing one particular pair $(x,y)$ the experimenter produces the state $x$ on which the measurement $y$ is performed. Consequently this experiment can be described by the experimentally accessible correlations $p(a | x,y)$ which give the probability of obtaining outcome $a$ when performing measurement $y$ on preparation $x$. It is important to note that the preparation as well as the measurement apparatuses are considered as a black box, that is, no assumptions are made about the state $\rho_x$ and the measurement description $\mathsf M_y$. This is with the exception that we assume that the dimension of the underlying Hilbert space is fixed. In this sense, properties of the system that can be deduced from the experimental data alone are semi-device-independent.

\subsection{Structure of the measurements}
Any $n$-outcome measurement $\mathsf M$ on a $d$-dimensional Hilbert space can be written as a POVM, that is, as positive semidefinite operators $\mathsf M= (M_{1},\dotsc,M_{n})$ satisfying $\sum_a M_{a} = \openone$. The operators $M_{a}$ are the effects of $\mathsf M$ and are associated with the outcomes $a \in \set{1,\dotsc,n}$ of $\mathsf M$. The set of all POVMs is convex, that is, it is closed with respect to taking probabilistic mixtures and efficient algorithms are known in order to decompose a given POVM into extremal POVMs \cite{A31}. By virtue of the Born rule, the probability $p(a)$ of obtaining outcome $a$ for a quantum state $\rho$ is given by $p(a) = \tr (\rho M_{a})$.

It is interesting to notice that a similar notion of POVMs can also be introduced in classical probability theory. Here, a general $n$-outcome measurement on a $d$-dimensional classical system is given by $n$ vectors from the $d$-dimensional unit cube $C_{d}=\set{\boldsymbol x\in \mathbb R^d| 0\le x_i\le 1}$ such that they sum up to $\boldsymbol{1}= (1,\dotsc,1) \in \mathbb{R}^{d}$. Therefore the set of all classical $n$-outcome measurements is equivalent to the set of all right-stochastic $d\times n$ matrices. We mention that this classical case is identical to the quantum case when one restricts all effects and states to be diagonal in some fixed basis.

When a measurement has only two nonzero outcomes then the measurement is dichotomic, for three nonzero outcomes it is trichotomic, and it is $n$-chotomic in the case of $n$ nonzero outcomes. An $n$-outcome POVM $\mathsf M$ can be simulated \cite{A15} with $n'$-chotomic POVMs $( \mathsf N_\ell )_{\ell}$ if there exists a probability distribution $(p_{\ell})_{\ell}$ such that $\mathsf M = \sum_{\ell} p_{\ell} \mathsf N_{\ell}$. Otherwise the measurement is irreducibly $n$-chotomic. For $n=3$ and $n'=2$ the simulation reduces to the randomization of three dichotomic measurements, that is,
\begin{equation}\label{eq:sim}\begin{split}
(M_{1},M_{2},M_{3}) = p_{1} (N_{1|1}, N_{2|1},0) +p_{2} (0,N_{2|2}, N_{3|2})\\
 +p_{3} ( N_{1|3},0, N_{3|3}),
\end{split}\end{equation}
where we wrote $N_{a|\ell}$ for the outcome $a$ of the measurement $\mathsf N_\ell$. These reducible three-outcome measurements form a convex subset of the set of all measurements. While in the $d$-dimensional classical probability theory, all measurements are reducible to $d$-outcome measurements, this is not the case in quantum theory \cite{Busch1995}. An archetypal counterexample is the trine POVM $\mathsf S$ which is composed out of three qubit effects $S_a=\frac23 \proj{S_a}$ where $\ket{S_a}$ for $a=1,2,3$ are located in a plane of the Bloch sphere and are rotated by an angle of $\frac23\pi$ against each other, see Figure~\ref{fig:structure1}.

\begin{figure}
\centering
\includegraphics[width=.8\linewidth]{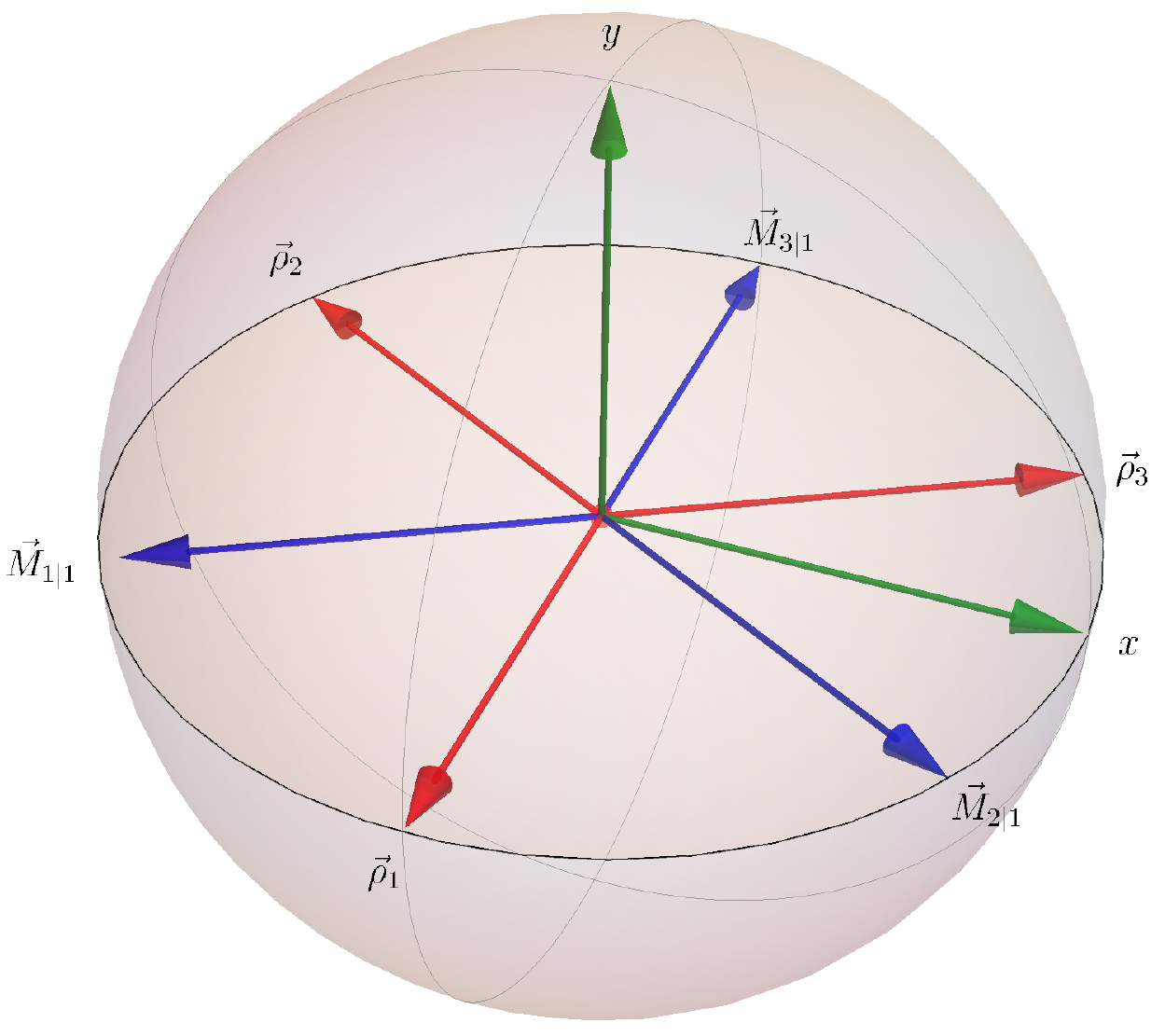}
\caption{\label{fig:structure1}%
Representation of the state- and effect configuration for generating trichotomic correlations with the trine POVM. The states $\rho_{x}$ are obtained from the vectors labeled by $\vec\rho_{x}$ via  $\rho_x=\frac12(\openone+\sum_k [\vec\rho_{x}]_k \sigma_k)$ for $x \in  \lbrace 1,2,3 \rbrace$, and the effects $M_{a|1}$ are obtained from the vectors labeled by $\vec M_{a|1}$ via $M_{a|1}=\frac{1}{3}(\openone+ \sum_k [\vec M_{a|1}]_k \sigma_k)$ for $a \in \lbrace 1,2,3 \rbrace$.}
\end{figure}

\subsection{Unambiguous state discrimination} \label{subsec:USD}
Unambiguous state discrimination \cite{A3} (USD) is a special instance of quantum state estimation. While the power of quantum state tomography relies on the access to a sufficiently high number of identically prepared quantum states, here only a single copy of the input state is available. In what follows we formulate the task of USD for the special case where the system subject to discrimination is prepared in one of two pure states. Then, one party (Alice) randomly but with equal probability chooses one of two states $\ket{\psi_{1}}$ and $\ket{\psi_{2}}$ which are known to both parties and sends it to a receiver (Bob). By measuring the incoming state, Bob must either correctly identify the state or declare that he does not know the answer, yielding an inconclusive result. Naturally, as soon as Alice's two states are not perfectly distinguishable, $\braket{\psi_1|\psi_2}\ne 0$, Bob cannot archive a unit success rate. It is well known \cite{A3} that for qubits Bob's best measurement is given by the irreducibly trichotomic POVM $\mathsf M= ( M_{1}, M_{2}, M_{3} )$ with
\begin{align} \label{usdpovm}
 M_{1} = \frac{\openone- \proj{\psi_{2}}}
         {1 + \abs{\braket{\psi_{1}|\psi_{2}}}}, \quad
 M_{2} = \frac{\openone- \proj{\psi_{1}}}
         {1 + \abs{\braket{\psi_{1}|\psi_{2}}}}
\end{align}
and $M_{3} = \openone- M_{1} -M_{2}$. With this construction one can see directly that if the measurement yields the outcome $a=1$ or $a=2$, then one can conclude that the received state was $\ket{\psi_{a}}$ while outcome $a=3$ does not allow a definite statement.

Motivated by the above considerations, we define now a family of correlations that is particularly useful for our later analysis. We consider a prepare-and-measure scenario with three preparations, $\rho_1$, $\rho_2$, and $\rho_3$, and two measurements, $\mathsf M_1 = ( M_{1|1}, M_{2|1}, M_{3|1} )$ and $\mathsf M_2 = ( M_{1|2}, M_{2|2} )$. The states and the measurements are chosen in such a way, that if $\mathsf M_2$ yields the outcome $1$ this implies that the received state was not $\rho_{1}$ and if the obtained state is $\rho_{2}$, then $\mathsf M_2$ produces with certainty outcome $2$. For $\mathsf M_1$ we impose that from outcome $1$ ($2$) it follows that the state was not $\rho_{3}$ ($\rho_{2}$). Since any POVM has to obey the normalization condition, the last effect can always be calculated from the previous ones. Hence we arrange the correlations $p(a|x,y)= \tr(\rho_{x} M_{a | y})$ in a $3\times 3$ matrix $\mathcal P$, where the rows correspond to the states and the columns to the effects $M_{1|1}, M_{2|1}, M_{1|2}$, that is,
\begin{align} \label{Eq:USDDistr}
\mathcal{P} =
\begin{pmatrix}
  p(1 | 1,1) & p(2 | 1,1) & 0 \\
  p(1 | 2,1) & 0 & 1 \\
  0              & p(2 | 3,1) & p(1 | 3,2)
\end{pmatrix}.
\end{align}
Since correlations of this form are motivated by USD, we refer to them as USD correlations. In particular, if the dimension of the system is known to be $d$, we write $\USD_d$ for the set of all USD correlation achievable under this constraint. Obviously the sets $\USD_{d}$ obey the inclusion $\USD_{d} \subset \USD_{d+1}$ and furthermore, as we point out in Section~\ref{minimal}, dimension 3 is already sufficient to achieve all USD correlations, $\USD_d=\USD_3$ for all $d>3$. We are therefore particularly interested in the qubit case, for which we have the following characterization.
\begin{lemma} \label{lem:para}
The set $\USD_2$ consists exactly of all correlations of the form
\begin{align} \label{eq:para}
 \mathcal{P} ( p,q,\xi) =
 \begin{pmatrix}
 p \xi    & q         & 0\\
 p(1-\xi) & 0         & 1\\
 0        & q (1-\xi) & \xi
 \end{pmatrix},
\end{align}
with $p,q,\xi \in [0,1]$ such that $(1-p)(1-q)\ge pq\xi$. In addition, for a given correlations matrix $\mathcal P$, the states and measurements realizing $\mathcal P$ are unique, up to a global unitary transformation.
\end{lemma}
\noindent
The proof is given in the Appendix~\ref{proof:lemma2}.

\section{Certification of trichotomic measurements} \label{Certification}
We approach now the question whether the difference between reducible and irreducible $n$-outcome measurements can be observed using only the correlations, that is, whether there exist an irreducible $n$-outcome POVM $\mathsf M=(M_1,\dotsc,M_n)$ together with auxiliary measurements $\mathsf M_2, \dotsc \mathsf M_m$ and states $\rho_1,\dotsc,\rho_s$ such the correlations $p(a|x,y)= \tr (\rho_x M_{a|y})$ cannot stem from a reducible measurement. Correlations of this type are genuinely $n$-chotomic, otherwise simulable $n$-chotomic. Clearly if such correlations exist, they enable us to certify that the measurement $\mathsf M_1$ is indeed irreducible. To point out the difference between both questions consider the following example. Suppose that one implements the trine POVM $\mathsf S$ on a qubit system, which is an irreducible three-outcome measurement. We show in Section~\ref{minimal} that this measurement alone can never yield correlations which cannot be explained by a reducible measurement. Therefore an irreducible three-outcome measurement does not necessarily define genuine trichotomic correlations while genuine trichotomic correlations always involve an irreducible measurement.

\subsection{Minimal scenario} \label{minimal}
Suppose we want to certify that a given measurement apparatus implements an irreducible $n$-outcome POVM $\mathsf M_1$. What is the minimal scenario, in which one can conclude from the output statistics alone that the POVM is irreducible? In other words, what is the minimal number of state preparations $s$ and auxiliary measurements $m-1$?

Clearly, if one has access to $s$ preparations and $m-1$ auxiliary measurements, the set of all possible correlations $p(a|x,y)$ that can be obtained in a scenario without any constraint on the dimension of the system, yields a convex set. As it turns out, this convex set is a polytope, whose extremal points correspond to deterministic correlations \cite{Hoffmann:2018NJP}, that is, where all $p(a|x,y)$ are either 0 or 1. If the dimension $d$ of the system is at least $s$ then these extremal points can be obtained from a fixed choice of $s$ orthogonal states $\rho_x=\proj{\psi_x}$ and at most $s$-chotomic measurements with effects $\mathsf M_{a|y}=\sum_x p(a|x,y)\rho_x$. Hence, all correlations can be written as convex combination of deterministic strategies. Since all extreme points use the same states, the convex coefficients can be absorbed into the effects yielding at most $s$-chotomic POVMs. For our case of an irreducible three-outcome measurement on a qubit, $n=3$ and $d=2$, this implies that at least $s=3$ different states are required. It also follows that $\USD_d\subset \USD_3$ since only three different states are used in the USD correlations.

Regarding the number of auxiliary measurements, $m-1$, we consider the case where no auxiliary measurements are used. In this scenario, we write $\mathcal C_d(s,n)$ for the convex hull of all correlations on a $d$-dimensional classical system and $\mathcal Q_d(s,n)$ for the convex hull of all correlations on a $d$-dimensional quantum system. Note that these sets are only of importance for determining the minimal scenario and are not used subsequently.
In Theorem~3 in Ref.~\cite{A7} it was established that both sets are equal, $\mathcal Q_d(s,n)=\mathcal C_d(s,n)$. Since all correlations in $\mathcal{C}_{d}(s,n)$ can be obtained from $d$-chotomic measurements, this property also follows for all correlations in $\mathcal{Q}_{d}(s,n)$. Therefore the smallest scenario which allows us to certify an irreducible three-outcome measurement on a qubit includes at least two measurements, $\mathsf M_{1} = (M_{1 | 1}, M_{2 | 1}, M_{3 | 1})$ and the auxiliary measurement $\mathsf M_{2} = (M_{1|2}, M_{2 | 2})$, see also Figure~\ref{fig:1}.

\begin{figure}
\centering
\includegraphics[width=0.9\linewidth]{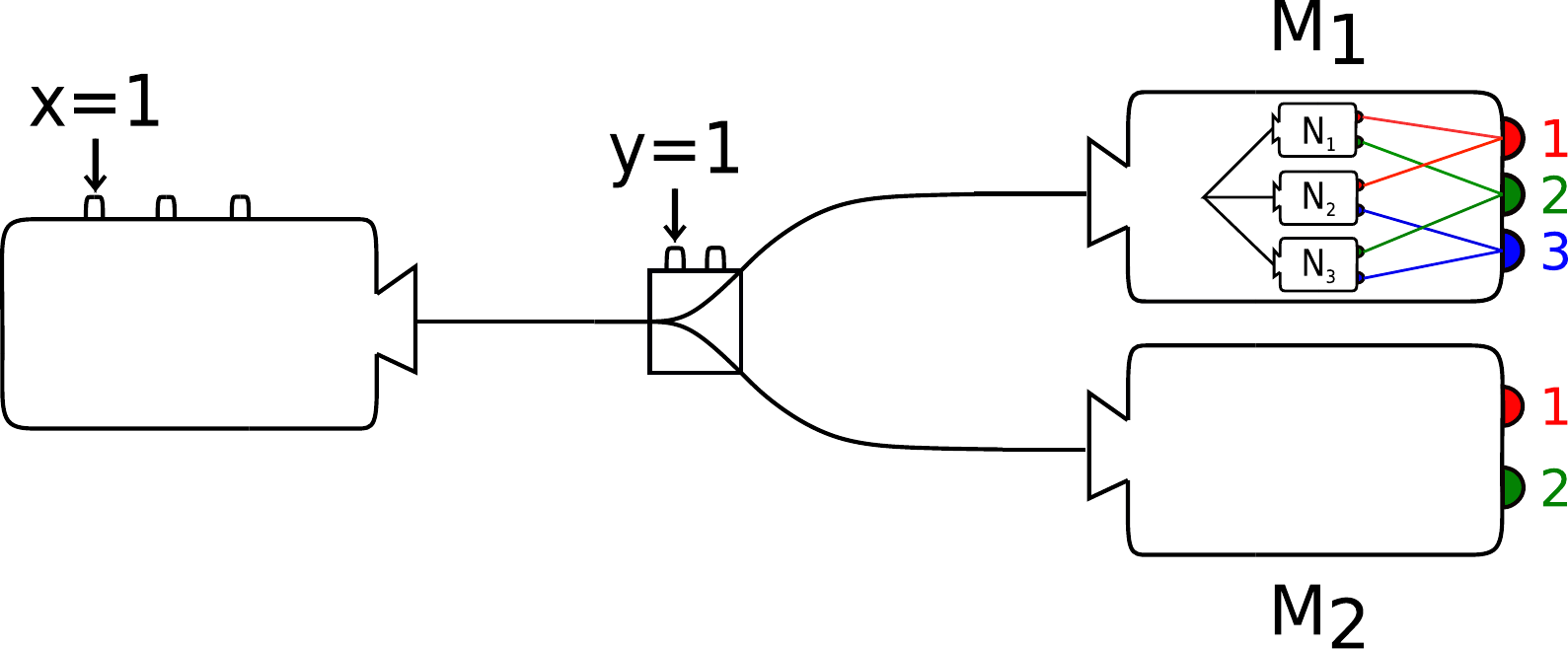}
\caption{\label{fig:1}%
Prepare-and-measure setup of the minimal scenario producing simulable trichotomic correlations. On the left hand side one can choose between three different preparation procedures $x \in \lbrace 1,2,3 \rbrace$ and the corresponding qubit state is sent to one of two measurement devices (right hand side). The measurement device is chosen by the experimenter and this choice is denoted by $y \in \lbrace 1,2 \rbrace$. For $y =1$ the measurement can yield three different outcomes, but governed by a specific inner mechanism. It consists of three two-outcome measurements, one of which is chosen at random. Depending on the outcome of this measurement, an outcome is assigned to the overall three-outcome measurement. If $y=2$ the measurement device is a simple two-outcome measurement.}
\end{figure}
The set of all correlations achievable in the minimal scenario with fixed dimension $d$ is subsequently denoted by $\COR_d$. Furthermore we write $\SIM_{d}$ for the set of simulable trichotomic correlations within $\COR_d$. For a concise description of all sets that are relevant for our discussion, see also Table \ref{tab:summary}. Clearly, the correlations $\USD_d$ are a subset of $\COR_d$ and $\SIM_{d} \cap \USD_{d} $ are the simulable trichotomic correlations within $\USD_d$.
\begin{table}[h]
\centering
\begin{tabular}{|l|c|}
\hline
Symbol & Meaning\\
\hline
$\USD_{d}$ & \makecell{ Set of correlations of the form (\ref{Eq:USDDistr}) that can \\ be obtained from a $d$-level quantum system. }\\
\hline
$\SIM_{d}$ &   \makecell{Set of simulable trichotomic correlations\\ that can be obtained from a $d$-level\\ quantum system in the minimal scenario.}\\
\hline
$\COR_{d}$ &  \makecell{Set of all correlations achievable\\ in the minimal scenario with fixed dimension $d$.}\\
\hline
\end{tabular}
\caption{Explanation of the meaning of the introduced symbols that are important for the subsequent paragraphs.}
\label{tab:summary}
\end{table}

\subsection{Geometry of the trichotomic correlations}
In this section we investigate how the sets $\SIM_{2}$ and $\COR_{2}$ are related. In contrast to Bell-like scenarios, where the convex hull of the correlations is taken \cite{A6,A12}, here we restrict to the bare sets $\COR_2$ and $\SIM_2$. These sets are not convex as can be seen by considering the subset $\USD_2$, characterized by Lemma~\ref{lem:para}. In fact, it is evident that the correlation matrix $D_0=\mathcal{P}(1,0,1)$ and $D_1=\mathcal{P}(1,1,0)$ can be realized with dichotomic measurements, that is, $D_0,D_1\in \USD_2\cap\SIM_2$. However, no convex combination $D_\lambda= \lambda D_{0} + (1- \lambda)D_{1}$ with $0<\lambda<1$ can be written in the form $\mathcal{P}(p,q,\xi)$ as given by Eq.~\eqref{eq:para}. Here it is important to note, that the correlations $D_{\lambda}$ are still USD correlations. Hence $D_{\lambda} \in \COR_{2}$ already implies $D_{\lambda} \in \USD_{2}$. The relation of the points $D_{1}, D_{2}$ and their convex mixture $D_{\lambda}$ to the sets $\USD_{2}$ and $\SIM_{2}$ is illustrated in Figure \ref{fig:illustration}. Since $\USD_2$ is an affine section of $\COR_2$, it follows that neither $\COR_2$ nor $\SIM_2$ are convex.

Our next step is to establish that not all qubit USD correlations are simulable trichotomic. We have the following theorem which we prove in Appendix~\ref{proof:thm2}.
\begin{theorem} \label{maintheorem}
There exist correlations in $\text{USD}_{2}$ that are not contained in the convex hull of $\text{USD}_{2} \cap \SIM_{2}$. 
\end{theorem}

\noindent
Hence, even the convex hull of the simulable trichotomic qubit USD correlations does not cover all qubit USD correlations. This might rise the expectation that there exists a linear inequality separating $\COR_2$ and $\SIM_2$, despite of the nonconvexity properties discussed above. But note, that the statement of Theorem~\ref{maintheorem} only concerns the subset of USD correlations and one cannot conclude that the convex hull of $\SIM_2$ is a proper subset of $\COR_2$.

\subsection{Genuine trichotomic correlations}\label{ssec:genuine}
For an experimental certification of an irreducible three-outcome measurement it is essential to find correlations $\mathcal P\in\COR_2$ such that closest simulable correlations $\mathcal P'\in \SIM_2$ have a distance $r$ of reasonable size. We measure this distance either in terms of the supremum norm, yielding $r_\infty$, or in terms of the Euclidean norm, yielding $r_2$, that is,
\begin{align}
\label{eq:uniform}
r_\infty&= \max_{i,j} \, \abs{\mathcal P_{i,j}-\mathcal P'_{i,j}},\\
\label{eq:euclid}
r_2&= \big[\sum_{i,j}(\mathcal P_{i,j}-\mathcal P'_{i,j})^{2} \big]^{\frac12}.
\end{align}
According to Theorem~\ref{maintheorem} we can preliminarily focus on the family of USD correlations, since Theorem~\ref{maintheorem} guarantees that there exist USD correlations $\mathcal P\in \USD_2$ such that $r>0$.

In order to compute $r_2$ and $r_\infty$ we rely on numerical optimization over the set $\SIM_2$. The optimization is nonlinear and involves three (possibly mixed) states as well as four dichotomic POVMs. It is important to note that if the states or the effects are fixed, the problem can be rephrased as a semidefinite program (SDP) and becomes thereby easy to solve numerically. We use this fact to implement an alternating optimization (``seesaw'' algorithm \cite{Pal2010}), a detailed description is provided in Appendix~\ref{optimizationUp}. We mention that, technically, this optimization algorithm only yields guaranteed upper bounds on the distances, because it is based on finding the correlations in $\SIM_2$ closest to $\mathcal P$. While this fact is not of practical concern, we discuss in Appendix~\ref{optimizationLow} also how strict, but rather rough, lower bounds can be obtained. The largest distance we found is realized by the choice of $\mathcal{P} = \mathcal{P} (0.577, 0.726,0.276)$, where the seesaw algorithm yields $r_{2} \approx 0.0391$ and $r_{\infty} \approx 0.0177$. Using the algorithm described in Appendix \ref{optimizationLow}, the computation of the lower bound yields $r_{\infty} \approx 0.0022$, hence one order of magnitude smaller. However, the lower bound should only emphasize that the distance is indeed positive, hence consistent with Theorem \ref{maintheorem}. Larger distances can be achieved using correlations which are not confined to $\USD_2$. In particular, choosing an arrangement involving the trine POVM, see Figure~\ref{fig:structure1} and Appendix~\ref{app:optimal}, we find $r_{2} \approx 0.0686$ and $r_{\infty} \approx 0.0342$. Hence we obtain only moderate experimental requirements for a certification of an irreducible three-outcome measurement.

\begin{figure}
\includegraphics[width=.7\linewidth]{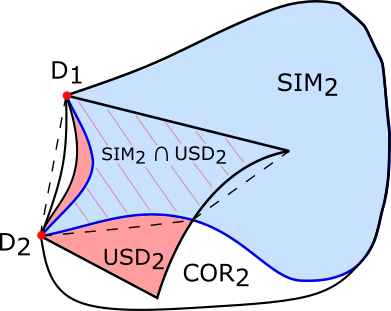}
\caption{ \label{fig:illustration}
Schematic illustration of the relations among the sets $\USD_{2}$, $\SIM_{2}$ and $\COR_{2}$. The set $\COR_{2}$ of all correlations that can be obtained with a qubit in the minimal scenario is obviously a superset of $\SIM_{2}$, the set of all simulable correlations in the minimal scenario as well as $\USD_{2}$, the set of all correlations that can be obtained with a qubit and are of the form (\ref{Eq:USDDistr}). The points $D_{1}$ and $D_{2}$ represent deterministic correlations contained in $\USD_{2} \cap \SIM_{2}$, hence located at the boundary of any of the three sets. The dashed line represents the convex hull of the set $\USD_{2} \cap \SIM_{2}$. }
\end{figure}

\section{Estimating the state space dimension} \label{Spekkens}
As discussed in the introduction, in order to certify an irreducible $n$-outcome measurement, knowledge of the dimension of the prepared system is necessary. While it is possible that the dimension can be convincingly deduced from the experimental setup, for a fully device-independent procedure, the dimension of the system has to be determined from correlation data alone. However, the dimension of a system is a physically ill-defined object, in the sense that any description of a system can always be embedded into higher dimensional system. For example, we can treat a qubit as a restricted theory of a qutrit. But from an operational point of view, one can still assess the dimension by determining the effective dimension, that is, the minimal dimension which explains the experimental data. A dimension witness \cite{A2} might seem to be the appropriate tool for this purposes, since it gives a procedure to determine a lower bound on the dimension. However, this is not sufficient for our purposes, because it does not exclude that the effective dimension can be higher than the dimension witnessed.

In general we assume that the procedure to determine the effective dimension of a system consists of $s$ different state preparations and $m$ measurements. The correlations $p(a|x,y)$ form a matrix $(\mathcal A_{x,\ell})_{x,\ell}$, where $x$ labels the states, $a$ labels the outcome of measurement $y$, and $\ell$ enumerates all outcomes of all measurements. For a $d$-level system, the rank of this matrix can be at most the affine dimension of the state space, that is, $d^2-1$. Hence determining the rank of the matrix $\mathcal A$ can give an estimate of the effective dimension of the system. In practice, one would choose a large number of preparation procedures and a large number of measurement procedures with the expectation, that an estimate of the rank of $\mathcal A$ produces a reliable estimate of the affine dimension of the state space. We mention that for consistently reasons, the preparation- and measurement procedures should include those required to certify the irreducibility of the $n$-outcome measurement.

While this approach can work in principle, it has to be considered with care. For an implementation of an irreducible three-outcome measurement on a qubit, it is typically necessary to dilate the three-outcome measurement to a projective measurement on a higher-dimensional system \cite{Peres:1993}. Despite of this, it still makes sense to speak about an irreducible three-outcome measurement if the additional dimensions used for the dilated measurement are not accessible due to a physical mechanism that reduces the dimension before entering the measurement station. In a setup using the polarization degree of freedom of a photon, this may be achieved, for example, by means of a single mode fiber. Mathematically, such a mechanism correspond to a completely positive map $\Phi$ so that the correlations are obtained through $p(a|x,y)=\tr[\Phi(\rho_x) M_{a|y}]$. However, then the rank of the matrix $\mathcal A$ alone is insufficient to establish the effective dimension of the system, as can be seen by considering the dephasing qutrit-qutrit channel $\Phi\colon \rho \mapsto \sum_k \proj k \rho \proj k$. Using this channel, the matrix $\mathcal A$ will have rank three, which would suggest an effective dimension of $d=2$, while the actual effective dimension is $d=3$. This can be overcome by certifying that the shape of the state- and effect space correspond to a qubit. For methods to implement such a certification we refer here to Ref.~\cite{A1}.

\section{Discussion} \label{conclusion}
We studied the structure of the correlations produced by irreducible three-outcome qubit measurements in the prepare-and-measure scenario. Our goal was a minimal scenario in terms of the number of experimental devices required. Using only one auxiliary measurement we found that the genuine trichotomic correlations and the simulable trichotomic correlations can be separated in the Euclidean norm by $r_2=0.0686$. Consequently, all 9 entries in the correlation matrix must be determined with an absolute error of roughly $r_2/\sqrt 9\approx 0.0229$. We mention the great similarity of this setup with the one used in Ref.~\cite{Tavakoli2020aSW}, namely the states and the measurements are the same, but in our setup we omit one of the auxiliary measurements. While this scheme is motivated by symmetry considerations, we also used the USD correlations to systematically describe a subset of the trichotomic correlations. But within these correlations the largest distance we obtained is smaller, $r_2\approx 0.0347$. We mention that the results in Ref.~\cite{Tavakoli2020bSW} are also based on USD, but with the different goal to certify any of several trichotomic measurement.

Our results are not based on a linear inequality separating genuine and simulable trichotomic correlations and such an inequality is not necessary for the purpose of an experimental certification. However, we also established in Theorem~\ref{maintheorem} that such an inequality exist when the analysis is constrained to the USD correlations. It is now an interesting open question whether in our minimal scenario this also holds when considering all correlations, since this would imply that also in the minimal scenario the convex hull of the simulable correlations does not include all trichotomic correlations. A further interesting question is if it is possible to find a systemic approach to construct families of distributions in a minimal scenario that are able to decide whether an implemented $n$-chotomic measurement is irreducible. Such an approach would unify recent results, that actual appear in separated contexts. 

\begin{acknowledgements}
We thank Rene Schwonnek, Xiao-Dong Yu and Eric Chitambar for discussions and the University of Siegen for enabling our computations through the HoRUS  and the OMNI cluster.
This work was supported by
the Deutsche Forschungsgemeinschaft (DFG, German Research Foundation, project numbers 447948357 and 440958198),
the Sino-German Center for Research Promotion (Project M-0294), and
the ERC (Consolidator Grant 683107/TempoQ). JS acknowledges support from the House of Young Talents of the University of Siegen.
\end{acknowledgements}

\appendix

\section{Proof of Lemma~\ref{lem:para}}\label{proof:lemma2}
Suppose that $\mathcal P \in \USD_{2} $. This has consequences for the measurements $\mathsf{M}_{1} = (M_{1 | 1}, M_{2 | 1}, M_{3 | 1})$, $\mathsf{M} _{2} = (M_{1 | 2}, M_{2 | 2})$ and the states $\rho_{1},\rho_2, \rho_{3}$ that can realize the correlations. In particular, $\mathcal P_{2,3}=\tr [\rho_{2} M_{1 | 2 } ] =1$ and $\mathcal P_{1,3}=0$ imply $\rho_{2} = M_{1 | 2} = | \varphi \rangle \langle \varphi \vert$ and $\rho_{1} = | \varphi^{\perp} \rangle \langle \varphi^{\perp} \vert$, where $\ket\varphi$ and $\ket{\varphi^\perp}$ are two orthonormal vectors. With a similar argument we obtain $M_{2 | 1} = q | \varphi^{\perp} \rangle \langle \varphi^{\perp} \vert$ with $0\le q\le 1$. It remains to consider the consequences of $\mathcal P_{3,1}=0$ for $\rho_{3}$ and $M_{1 | 1}$. This requires $M_{1 | 1} = p | \eta^{\perp} \rangle \langle \eta^{\perp}\vert$ and $\rho_{3} = | \eta\rangle \langle \eta \vert$ for some orthonormal vectors $\ket\eta$ and $\ket{\eta^\perp}$ and $0\le p\le 1$.

Without loss of generality, we can assume
\begin{align}
 \ket\eta&= \sqrt{\xi}\ket\varphi+\sqrt{1-\xi}e^{i\phi}\ket{\varphi^\perp}
 \text{ and}\\
 \ket{\eta^\perp}&= \sqrt{1-\xi}\ket\varphi-\xi e^{i\phi}\ket{\varphi^\perp}
\end{align}
where $0\le \xi\le 1$ and $\phi\in\mathbb R$. This yields immediately Eq.~\eqref{eq:para} together with the conditions $p,q,\xi\in[0,1]$. For $\mathsf M_1$ to form a POVM it remains to verify that $M_{3|1}=\openone-M_{1|1}-M_{2|1}$ is positive semidefinite. This reduces here to $\tr(M_{3|1})\ge0$ and $\det(M_{3|1})\ge0$ and can be equivalently expressed as the single condition $(1-p)(1-q)\ge pq\xi$.

From the above construction it is also immediately clear that conversely, any choice of $p,q,\xi$ satisfying the constraints in Lemma~\ref{lem:para} is in $\USD_2$. Finally, given $\mathcal P$, all effects and states are fixed by the above considerations, except for the choice of the orthonormal basis $\{\ket\varphi,e^{i\phi}\ket{\varphi^\perp}\}$, proving the claim of a unique representation up to a unitary transformation.


\section{Proof of Theorem~\ref{maintheorem}} \label{proof:thm2}

We first parameterize the correlations $\mathcal D\in \SIM_{2} \cap \USD_{2}$ with $\mathcal D_{3,3}\ne 0$. By virtue of Eq.~\eqref{eq:sim} and using the proof of Lemma~\ref{lem:para} above, the effects of the simulated trichotomic POVM $\mathsf M_1$ can be written as
\begin{align}
\label{SIP:eq1}
M_{1|1}&=p\proj{\eta^\perp}= \kappa_1 F_1+\kappa_3(\openone-F_3) \text{ and}\\
\label{SIP:eq2}
M_{2|1}&=q\proj{\varphi^\perp}= \kappa_1 (\openone-F_1) +\kappa_2 F_2,
\end{align}
where, compared to Eq.~\eqref{eq:sim}, we wrote $\kappa_j$ in place of $p_j$ and $F_j$ in place of $N_{j|j}$. For $\kappa_1\ne 0$ it follows that $F_1\propto \proj{\eta^\perp}$ and $\openone-F_1\propto \proj{\varphi^\perp}$, yielding $F_1=\proj{\eta^\perp}=\proj{\varphi}$. This is in contradiction to the assumption $\mathcal D_{3,3}\ne 0$ by virtue of $\mathcal D_{3,3}=\tr(M_{1|2}\rho_3)$ with $\rho_3=\proj \eta$ and $M_{1|2}=\proj{\varphi}$. Hence $\kappa_1=0$. Writing $\kappa\equiv \kappa_2=1-\kappa_3$, Eq.~\eqref{SIP:eq2} implies $F_2=f_2\proj{\varphi^\perp}$ with $0\le f_2\le 1$ and $\kappa f_2=q$. Similarly, from Eq.~\eqref{SIP:eq1} one obtains $\openone- F_{3} = f_{3} \proj{\eta^{\perp}}$ with $(1-\kappa)f_3=p$. Therefore, $\mathcal D\in \USD_2\cap\SIM_2$ with $\mathcal D_{3,3}\ne 0$ if and only if
\begin{align}\label{eq:Dmat}
 \mathcal D =
    \begin{pmatrix}
     f_{3} (1- \kappa) \xi &  f_{2} \kappa & 0 \\
     f_{3} (1- \kappa) (1-\xi)& 0 & 1 \\
     0 & f_{2} \kappa (1-\xi) & \xi
    \end{pmatrix}
\end{align}
with $f_2,f_3,\kappa,\xi\in [0,1]$ and $\xi\ne 0$.

Next we show that for any $0<\xi<1$ there exist correlations $\mathcal P \in \USD_{2}\setminus \SIM_2$. For this we consider the linear map
\begin{equation}
W\colon \mathcal P\mapsto -\mathcal P_{1,1}-\mathcal P_{1,2}+\mathcal P_{3,2}+\mathcal P_{3,3}.
\end{equation}
For $\mathcal P(p,q,\xi)$ as in Eq.~\eqref{eq:para} and the choice $p=\frac12$ and $q=1/(\xi+1)$, one verifies that for any $0<\xi<1$ the constraint $(1-p)(1-q)\ge pq\xi$ is satisfied and in addition $W(\mathcal P)< 0$ holds.
However, for any $\mathcal D$ as in Eq.~\eqref{eq:Dmat}, we have $W(\mathcal D)=
 \xi[1-f_2\kappa-f_3(1-\kappa)]\ge 0$ and in addition for any $\mathcal T\in \USD_2$ with $\mathcal T_{3,3}=0$ we have immediately $W(\mathcal T)=0$.
Consequently, $W(\mathcal S-\mathcal P)>0$ for our choice of $\mathcal P$ and any $\mathcal S\in \SIM_2\cap\USD_2$, that is, $\mathcal S\ne \mathcal P$.

From this observation, it readily follows that neither does the convex hull of $\USD_2\cap\SIM_2$ contain all of $\USD_2$. This is the case, because $W$ is a linear map and thus its minimum over the convex hull of $\USD_2\cap \SIM_2$ is attained already for some $\mathcal S\in \USD_2\cap \SIM_2$. But for this set we just proved that $W(\mathcal S)>W(\mathcal P)$ for certain $\mathcal P\in \USD_2$.

\section{Upper bound on the minimal distance} \label{optimizationUp}
We compute the maximal radius $r$ of a ball $B_{r}(\mathcal P)$ around correlations $\mathcal P\in \COR_2\setminus \SIM_2$ such that $B_{r}(\mathcal P) \cap \SIM_{2}$ is empty. As we mention in the main text, we consider this maximal ball with respect to the supremum norm $r=r_\infty$ and the Euclidean norm $r=r_2$, see Eq.~\eqref{eq:uniform} and Eq.~\eqref{eq:euclid}, respectively. Correspondingly, computing $r$ can be formulated as the optimization problem to minimize a real parameter $t$ over all $\mathcal Q\in \SIM_2$, such that
\begin{equation}\label{eq:consinfty}
 -t \le \mathcal P_{i,j}-\mathcal Q_{i,j} \le t\quad \forall i,j
\end{equation}
in case of $r_\infty$ and in the case of $r_2$,
\begin{equation}\label{eq:cons2}
 \sum_{i,j} (\mathcal P_{i,j}-\mathcal Q_{i,j})^2 \le t^2.
\end{equation}

We write $F_1=M_{1|1}$, $F_2=M_{2|1}$, $F_3=M_{1|2}$ so that $Q_{i,j}=\tr(\rho_i F_j)$. If we keep the effects fixed, then the optimization is a semidefinite program of the following type: Minimize $t$ under the constraints $\rho_i\ge0$ and $\tr(\rho_i)=1$ for $i=1,2,3$ and either the linear constraint \eqref{eq:consinfty} or the quadratic-convex constraint \eqref{eq:cons2}. Similarly, if we keep the states fixed, then the optimization is again a  semidefinite program, however now with the constraints on the states replaced by constraints on the effects, namely,
\begin{gather}
 F_1 = F'_1+F'_0,\quad F_2= F'_2 + f_0 \openone - F'_0,\\
 0\le F'_0\le f_0\openone, \quad
 0\le F'_1\le f_1\openone, \quad
 0\le F'_2\le f_2\openone,\\
 0\le F_3\le \openone, \quad
 f_0+f_1+f_2=1.
\end{gather}
Since these small semidefinite programs can be solved very fast numerically, this invites for a seesaw optimization \cite{Pal2010} where one alternates between the two optimizations until $t$ converges.

We implement this seesaw algorithm using the Python library PICOS with the CVXOPT back-end.
 As criterion for convergence we take $t_{n-1}-t_{n}<10^{-6}$, where $t_n$ is the value after $n$ seesaw iterations. This convergence happens after at most 300 iterations. We repeat the optimization 4500 times, each time with different starting values for $\rho_1$, $\rho_2$, and $\rho_3$, where we take pure states chosen randomly according to the Haar measure and then decrease the purity $\tr(\rho^2)$ to be uniformly in the interval $[\frac 12,1]$. The same optimal value is always reached independently of the start values for the Euclidean norm, while it occurs only for about 1\% of the start values in the case of the supremum norm.

\section{Lower bound on the minimal distance} \label{optimizationLow}
In this appendix we describe an algorithm that yields a lower bound on the distance between a given correlation and the set $\SIM_{2}$. In the following we will write $B (P,\epsilon)$ for the ball at $P \in \text{COR}_{2}$ of radius $\epsilon$ in the metric induced by the $\infty$-norm. Considering a given correlation $P \in \USD_{2}$ of the form in Eq.~\eqref{Eq:USDDistr}, we are interested in the condition on the states and the measurements under which we have $Q \in B (P,\epsilon)$ for $\epsilon >0$. Clearly any $Q \in \COR_2$ is of the form
\begin{align} \label{ap:eq1}
Q=
\begin{pmatrix}
\tr (\rho_1 M_{1|1}) & \tr(\rho_1 M_{2|1}) & \tr(\rho_1 M_{1|2}) \\
\tr (\rho_2 M_{1|1}) & \tr(\rho_2 M_{2|1}) &\tr(\rho_2 M_{1|2}) \\
\tr (\rho_3 M_{1|1}) & \tr(\rho_3 M_{2|1})&\tr(\rho_3 M_{1|2}) \\
\end{pmatrix}
\end{align}
However, it turns out to be useful to group the parameters in (\ref{ap:eq1}) into two different classes. Denote $a=(\rho_1,\rho_2,M_{1|2})$ and $b=(\rho_3,M_{1|1},M_{2|1})$, which together form the full set of parameters of the distribution $Q = Q(a,b)$.
The question, how large can $\epsilon$ be chosen, such that there is no $Q(a,b) \in B(P,\epsilon)$ if the trichotomic measurement $\mathsf{M}_{1} = (M_{1|1},M_{2|1},M_{3|1})$ is limited to be simulable by dichotomic measurements. Observe that if $Q(a,b)$ are USD correlations, then the elements $Q_{13}$, $Q_{22}$, $Q_{31}$, and $Q_{23}$ all are $0$ or $1$ and this imposes strong constraints on the parameters $a$. In fact, if $Q(a,b)=P$, then the value of $a$ is fixed up to a unitary transformation, as we have seen in Appendix~\ref{proof:thm2}. We call this fixed value $a^\ast$.
In the following, we will extend the argument in Appendix~\ref{proof:thm2} to show that $Q(a,b) \in B(P,\epsilon)$ implies that  $a \in B(a^\ast,\mathcal{O}(\sqrt{\epsilon}))$.
Then, using the continuity of $Q$, one can show that $a \in B(a^\ast,\mathcal{O}(\sqrt{\epsilon}))$ implies $Q(a,b) \in B(Q(a^\ast,b),\mathcal{O}(\sqrt{\epsilon}))$.
Since $Q(a,b) \in B(P,\epsilon)$ and $Q(a,b) \in B(Q(a^\ast,b),\mathcal{O}(\sqrt{\epsilon}))$, it follows that $Q(a^\ast,b) \in B(P,\epsilon + \mathcal{O} (\sqrt{\epsilon}))$ by the triangular inequality.
We have thus reduced the problem of asking for the existence of $a$ and $b$, such that $Q(a,b) \in B(P,\epsilon)$ to simply asking for the existence of $b$ such that $Q(a^\ast,b) \in B(P,\epsilon + \mathcal{O} (\sqrt{\epsilon}))$.
The latter means, for a fixed value of $\epsilon$, and fixed values of $\rho_1$, $\rho_2$ and $M_{1|2}$, we ask for the feasibility of $\rho_3$, $M_{1|1}$ and $M_{2|1}$ such that $Q \in B(P, \epsilon+ O(\sqrt{\epsilon}))$. 
This is not yet an SDP, it can however be decomposed into a finite number of SDP's by scanning the values of $\rho_3$ and bounding the error in the finite scanning.

The detailed calculation of the above program is straightforward, yet, cumbersome. We will first introduce some notation. For convenience we parameterize the effects and the states by the Bloch coordinates. More precisely, for an effect $E$, we write $(x_{0}, \vec{x})$ which means $E = \frac{1}{2} \sum_{i=0}^{3} x_{i} \sigma_{i}$, where $\sigma_{0} := \openone$ and $\sigma_{i}$ are the Pauli matrices. Similarly, for a state $\rho$ we write $\vec y$ with $\rho=\frac12 (\openone+\sum_{i=1}^3 y_i \sigma_i)$.
In particular, we write $\rho_i =( 1, \vec{r}_i )$, $M_{1|1} = ( x_{01}, \vec{x_1})$, $M_{2|1} = ( x_{02} \, , \vec{x_2})$, $M_{1|2} = (y_0, \vec{y})$. Because of the unitary freedom, we can always assume $\vec{x}_1$ and $\vec{x}_2$ to have no $z$-component and at the same time $\vec{r}_2$ to have no $y$-component. We also write $Q \in P \pm \epsilon$ as a synonym of $Q \in B(P,\epsilon)$.

Firstly, we show that, one can bound the trace the effect from below, given a lower bound on the probability for that effect.
\begin{lemma} \label{ap:lem1}
Let $\rho=(1, \vec{r} )$ be a state and $F=( x_0, \vec{x} )$ be an effect, we have
\begin{itemize}
\item[(i)] if $\tr [F \rho] \ge a$ then $x_0 \ge a$.
\item[(ii)] if $ \tr [F \rho] \le b$ then $x_0 \le 1+b$.
\end{itemize}
\end{lemma}
\begin{proof}
We first show (i). One has $\tr [F \rho] = \frac{1}{2} (x_0 + \vec{x} \vec{r}) \ge a$. Because $\vec{x} \vec{r} \le \abs{\vec{x}} \abs{\vec{r}} \le x_0 \abs{\vec{r}} \le x_0$, we have $x_0 + x_0 \ge 2 a$, or $x_0 \ge a$. For (ii), one has similarly, $\tr [F \rho] = \frac{1}{2} (x_0 + \vec{x} \vec{r}) \le b$, then $\frac{1}{2} (2-x_0 - \vec{x} \vec{r}) \ge 1-b$. Then applying (i), we find $2-x_0 \ge 1-b$, or $x_0 \le 1+b$.
\end{proof}

The following corollaries is an application of the lemma to the condition $Q \in P \pm \epsilon$.
\begin{corollary}[Estimation of the traces of effects]
\label{ap:lcor1}
For $Q \in P \pm \epsilon$ one needs
\begin{itemize}
\item[(i)] $1 - \epsilon \le y_0 \le 1+ \epsilon $
\item[(ii)] $c_1 - \epsilon \le x_{01} \le 1 + \epsilon$ with $c_1= \max\{P_{12},P_{22}\}$.
\item[(iii)] $c_2 - \epsilon \le x_{02} \le 1 + \epsilon$ with $c_2= \max\{P_{13},P_{33}\}$.
\end{itemize}
\end{corollary}
\begin{proof}
(i) is the direct consequence of $0 \le Q_{13} \le \epsilon$, and $1 \ge Q_{23} \ge 1 - \epsilon$. (ii) is the direct consequence of $Q_{11} \in P_{11} \pm \epsilon$ and $Q_{21} \in P_{21} \pm \epsilon$. Likewise, (iii) is the direct consequence of $Q_{12} \in P_{12} \pm \epsilon$ and $Q_{32} \in P_{32} \pm \epsilon$
\end{proof}

\begin{corollary}[Estimation of $\vec{r}_1$ and $\vec{r}_2$ by $\vec{y}$]
\label{ap:lcor2}
{\ \\ }
\begin{itemize}
\item[(i)] $Q_{23} \ge 1 -\epsilon$ implies $\abs{\vec{y}-\vec{r}_2} \le \sqrt{4 \epsilon + \epsilon^2}$.
\item[(ii)] $Q_{13} \le \epsilon$ implies $\abs{\vec{y}+\vec{r}_1} \le \sqrt{4 \epsilon + \epsilon^2}$.
\end{itemize}
\end{corollary}
\begin{proof}
(i) One has $Q_{23}= \frac{1}{2} (y_0 + \vec{y} \vec{r_2}) \ge 1 - \epsilon$, then $\vec{y} \vec{r_2} \ge 2(1-\epsilon)- y_0$.
Further, $\vec{y} \vec{r}_2 \ge 2(1-\epsilon)- y_0$ leads to $\vec{y}^2 + \vec{r}_2^2 - (\vec{y} - \vec{r}_2)^2 \ge 4 (1 -\epsilon) - 2 y_0$. Because $\abs{\vec{y}} \le \min\{y_0,2-y_0\}$, $\vec{r}_2^2 \le 1$, one obtains $(\min\{y_0,2-y_0\})^2 + 2 y_0 -3 + 4 \epsilon \ge (\vec{y} - \vec{r}_2)^2$. The left hand side is maximized at $y_0=1+\epsilon$, which leads to $4 \epsilon + \epsilon^2 \ge (\vec{y} - \vec{r}_2)^2$. The proof of (ii) is analogue to (i).
\end{proof}
Now we consider the constraint of the form $\tr [\rho F] \le \epsilon$. We show that if the trace of $F$ is bounded from below, we can bound the purity of $\rho$. This also extends to the consideration of each Bloch components of the Bloch vectors.
\begin{lemma}\label{ap:lem2}
Let $\rho=( 1,\vec{r} )$ be a state and $F=( x_0, \vec{x})$ be an effect with $x_0 \ge c$ and $\vec{x}$ having no $z$-component. Then $\tr[\rho F] \le \epsilon$ implies $\abs{\vec{r}} \ge \abs{\vec{r}_{xy}} \ge 1 - \frac{2 \epsilon}{c}$,
and $\abs{r_z} \le \sqrt{1-(1-2\epsilon/c)^2}$. As a consequence, if one lets $\vec{n}$ to be the unit vector of the direction of $\vec{r}_{xy}$, then $\abs{\vec{r}_{xy} -\vec{n}} \le 2 \epsilon/c$ and $\abs{\vec{r} - \vec{n}} \le 2 \epsilon/c + \sqrt{1-(1-2\epsilon/c)^2}$. Here $\vec{r}_{xy}$ denotes the projection of $\vec{r}$ onto the $xy$-plane and $\vec{r}_{z}$ onto the $z$-axis.
\end{lemma}
\begin{proof}
Observe that $\tr[\rho F]= \frac{1}{2} (x_0 + \vec{r} \vec{x}) = \frac{1}{2} (x_0 + \vec{r}_{xy} \vec{x}) \le \epsilon$. This leads to $-\vec{r}_{xy} \vec{x} \ge x_0 - 2 \epsilon$. Then we have $\abs{\vec{r}_{xy}} x_0 \ge x_0 - 2 \epsilon$ and so $\abs{\vec{r}} \ge \abs{\vec{r}_{xy}} \ge 1 - 2 \epsilon/x_0 \ge 1 - 2 \epsilon/c$. Also $1 \ge \vec{r}_z^2 + \vec{r}_{xy}^2 \ge \vec{r}_{z}^2 + (1 - 2 \epsilon/c)^2$, so $\abs{\vec{r}_z} \le \sqrt{1 - (1 - 2 \epsilon/c)^2}$.
The latter part of the statement is obvious.
\end{proof}

\begin{corollary}[Estimation of $\vec{r}_2$ and $\vec{r}_3$ by the unit vectors of their projection onto the $xy$-plane]
Let $\vec{r}_{2xy}$ denote the $xy$-component and $\vec{r}_{2z}$ denote the $z$-component of $\vec{r}_{2}$. Further let $\vec{n}$ be the unit vector in direction of $\vec{r}_{2xy}$ and $\vec{t}$ be the unit vector in direction of $\vec{r}_{3xy}$. Then we obtain as a direct consequence of Lemma~\ref{ap:lem2}
\begin{itemize}
\item[(i)] $Q_{31} \le \epsilon$ implies $\abs{\vec{r}_{3xy} - \vec{t}} \le 2 \epsilon/(c_1 - \epsilon) = \epsilon_{3xy}$ and $\abs{\vec{r}_{3z}} \le \sqrt{1-[1-2\epsilon/(c_1-\epsilon)]^2} = \epsilon_{3z}$.

\item[(ii)] $Q_{22} \le \epsilon$ implies $\abs{\vec{r}_{2xy} - \vec{n}} \le 2 \epsilon/(c_2 - \epsilon) = \epsilon_{2xy}$ and $\abs{\vec{r}_{2z}} \le \sqrt{1-[1-2\epsilon/(c_2-\epsilon)]^2} = \epsilon_{2z}$.
\end{itemize}
\end{corollary}
Now let us consider the cost of pinning the values of $\vec{r}_1$, $\vec{r}_2$ and $\vec{r}_3$. It is best to start with $\vec{r}_2$. We consider the estimation of $\vec{r}_2$ by $\vec{n}$. This leads to new estimation of the matrix elements involving $\rho_1$ and $\rho_2$. In particular we have
\begin{widetext}
\begin{align}
Q_{11} &= \frac{1}{2} \left( x_{01} + \vec{r}_1 \vec{x}_1 \right) \nonumber \\
&=\frac{1}{2} \left[ x_{01} + (-\vec{n}) \vec{x}_1 + (\vec{n}-\vec{r}_{2xy}) \vec{x}_1 + (\vec{r}_2 - \vec{y}) \vec{x}_1 + (\vec{y} +\vec{r}_1) \vec{x}_1\right] \nonumber \\
&\in \frac{1}{2} \left[ x_{01} + (-\vec{n}) \vec{x}_1 \pm (\epsilon_{2xy} + 2 \sqrt{4 \epsilon + \epsilon^2})\right]
\label{eq:errorfist}
\end{align}
\begin{align}
Q_{12} &= \frac{1}{2} \left( x_{02} + \vec{r}_1 \vec{x}_2 \right) \nonumber \\
& = \frac{1}{2} \left[ x_{02} + (-\vec{n}) \vec{x}_2 + (\vec{n}-\vec{r}_{2xy}) \vec{x}_2 + (\vec{r}_2 - \vec{y}) \vec{x}_2 + (\vec{y} +\vec{r}_1) \vec{x}_2\right] \nonumber \\
& \in \frac{1}{2} \left[ x_{02} + (-\vec{n}) \vec{x}_2 \pm (\epsilon_{2xy} + 2 \sqrt{4 \epsilon + \epsilon^2})\right]
\end{align}
\begin{align}
Q_{21} &= \frac{1}{2} \left( x_{01} + \vec{r}_2 \vec{x}_1 \right)
= \frac{1}{2} \left[ x_{01} + \vec{n} \vec{x}_1 + (\vec{r}_{2xy} - \vec{n})\vec{x}_1 \right]
\in \frac{1}{2} \left[ x_{01} + \vec{n} \vec{x}_1 \pm \epsilon_{2xy} \right]
\end{align}
\begin{align}
Q_{22} &= \frac{1}{2} \left( x_{02} + \vec{r}_2 \vec{x}_2 \right)
= \frac{1}{2} \left[ x_{02} + \vec{n} \vec{x}_2 + (\vec{r}_{2xy} - \vec{n})\vec{x}_2 \right]
\in \frac{1}{2} \left[ x_{02} + \vec{n} \vec{x}_2 \pm \epsilon_{2xy} \right]
\end{align}
Let us consider the error one introduces by fixing $\vec{r}_{3xy}$ to $\vec{t}$. This only affects the last row of the correlation table
\begin{align}
Q_{33} &= \frac{1}{2} \left( y_0 + \vec{y} \vec{r}_3\right) = \frac{1}{2} \left[ y_0 + \vec{r_2} \vec{r}_3 + (\vec{y} - \vec{r}_2) \vec{r}_3 \right] \nonumber
\in \frac{1}{2} [y_0 + \vec{r}_2 \vec{r}_3 \pm \sqrt{4 \epsilon + \epsilon^2}] \nonumber \\
& \in \frac{1}{2} [y_0 + \vec{r}_{2xy} \vec{r}_{3xy} + \vec{r}_{2z} \vec{r}_{3z} \pm (\sqrt{4 \epsilon + \epsilon^2})] \nonumber \\
& \in \frac{1}{2} [y_0 + \vec{n} \vec{t} + \vec{n} (\vec{r}_{3xy}-\vec{t}) + (\vec{r}_{2xy} - \vec{n}) \vec{t} + (\vec{r}_{2xy} - \vec{n})(\vec{r}_{3xy} - \vec{t}) \pm (\epsilon_{2z} \epsilon_{3z} + \sqrt{4 \epsilon + \epsilon^2})] \nonumber \\
& \in \frac{1}{2} [y_0 + \vec{n} \vec{t} \pm (\epsilon_{2xy} + \epsilon_{3xy} + \epsilon_{2xy} \epsilon_{3xy} + \epsilon_{2z} \epsilon_{3z} + \sqrt{4 \epsilon + \epsilon^2})].
\end{align}
\begin{align}
Q_{31} & = \frac{1}{2} \left( x_{01} + \vec{x}_1 \vec{r}_3 \right)
= \frac{1}{2} \left( x_{01} + \vec{x}_1 \vec{r}_{3xy} \right) \nonumber
 = \frac{1}{2} \left[ x_{01} + \vec{x}_1 \vec{t} + \vec{x}_1 (\vec{r}_{3xy} - \vec{t})\right] \nonumber \\
& \in \frac{1}{2} \left[ x_{01} + \vec{x}_1 \vec{t} \pm \epsilon_{3xy} \right]
\end{align}
\begin{align}
Q_{32} & = \frac{1}{2} \left( x_{02} + \vec{x}_2 \vec{r}_3 \right)
 = \frac{1}{2} \left( x_{02} + \vec{x}_2 \vec{r}_{3xy} \right)
 = \frac{1}{2} \left[ x_{02} + \vec{x}_2 \vec{t} + \vec{x}_2 (\vec{r}_{3xy} - \vec{t})\right] \nonumber \\
& \in \frac{1}{2} \left[ x_{02} + \vec{x}_2 \vec{t} \pm \epsilon_{3xy} \right] \label{eq:errorlast}
\end{align}
\end{widetext}
To sum up, as long as $Q \in P \pm \epsilon$, the Eqns.~\eqref{eq:errorfist}--\eqref{eq:errorlast} should be satisfied. This further implies the constraints on the values on right hand sides of Eqns.~\eqref{eq:errorfist}--\eqref{eq:errorlast}. For example, Eq.~\eqref{eq:errorfist} together with $Q_{11} \in P_{11} \pm \epsilon$ implies $1/2[x_{01} + (\vec{n} \vec{x}_1)] \in P_{11} \pm [\epsilon + (\epsilon_{2xy} + 2 \sqrt{4 \epsilon + \epsilon^2})]$, and so on. 

Therefore, given $\epsilon >0$ arbitrary, we have to ask whether there exist feasible $\vec{n}$ and $\vec{t}$ and $(x_{10},\vec{x}_1)$, $(x_{20},\vec{x}_2)$ which are simulable by dichotomic measurements such that all these constraints are satisfied.

Notice, that due to the unitary freedom mentioned at the beginning of this section, one can always fix $\vec{n}=(1,0,0)$. If we further fix $\vec{t}= (\cos (\varphi), \sin (\varphi), \ 0)$, asking for the existence of $(x_{10},\vec{x}_1)$, $(x_{20},\vec{x}_2)$ which are simulable by dichotomic measurements is an SDP.

To conclude, we have the following algorithm. Scanning over $\vec{t}= (\cos (\varphi), \sin(\varphi), 0)$ with certain finite step in $\varphi$, for any value of $\vec{t}$, fix $\vec{r}_2=\vec{n}$, $\vec{r}_1=-\vec{n}$, $\vec{y} = \vec{n}$ and $\vec{r}_3=\vec{t}$. Fix a value $\epsilon$ and test whether the SDP of finding reducible $M_{1|1}$, $M_{2|1}$ such that all the above-mentioned constrains are satisfied. At $\epsilon=0$, the SDP is infeasible. One can implement a bisection method to find out the exact transition point where the SDP is infeasible. We obtain the critical error tolerance $\epsilon_c(\varphi)$. By scanning all values of $\varphi$, one finds $\epsilon^*= \min \epsilon_c(\varphi)$.

For the practical purpose, the above described procedure is sufficient. In principle, one may still object that the number $\epsilon^*$ may not be reliable because of the finite step scanning over the values of $\varphi$. This objection can also be addressed. The idea is that the error introduced by the finite steps can in fact be bounded by bounding the variation of the function $\epsilon_c(\varphi)$. That is, we can find a number $C$ such that $\abs{\epsilon_c(\varphi + x)- \epsilon_c(\varphi)} \le C \delta$ for all $\varphi \in [0, 2\pi]$ and $\abs{x} \le \delta$. Note that the variation of $\varphi$ only affects $Q_{3k}$. A variation of $\delta$ in $\varphi$ gives rise to a variation of $\delta \vec{t}$ with $\abs{\delta \vec{t}} \le \delta$. One then sees that $\delta Q_{3k} \le \delta/2$. Therefore $\abs{\epsilon_c(\varphi + x)- \epsilon_c(\varphi)} \le \delta/2$ for any value of $\varphi$ and $\abs{x} \le \delta$. If one selects a step in $\varphi$ with size $\delta$, the error in the global minimum $\epsilon^* = \min \epsilon_c(\varphi)$ is bounded by $\delta/4$ (since the maximum distance from any point to a computed point is $\delta/2$). Taking $\delta=2 \pi \times 10^{-5}$ is sufficient to bound the error by $\pi/2 \times 10^{-5}$. An adaptive scheme of varying the step sizes over different regimes of $\varphi$ can be utilized to speed up the computation.

\section{Correlations from the trine POVM} \label{app:optimal}
At the end of Section~\ref{ssec:genuine} and in Figure~\ref{fig:structure1} we use an arrangement of states and effects involving the trine POVM. This arrangement produces the genuine trichotomic correlations
\begin{equation}
   \mathcal P_{\text{trine}} =
    \begin{pmatrix}
\frac{1}{2}&  \frac{1}{2} &  0\\
\frac{1}{2}&   0 & \frac{3}{4}\\
0          &   \frac{1}{2} &\frac{3}{4}
    \end{pmatrix}
\end{equation}
 with an Euclidean distance of $r_2\approx 0.0686$ to the simulable trichotomic correlations. The correlations $\mathcal P_\text{trine}$ are obtained with the states
\begin{equation}\begin{split}
 \rho_{1} &= \frac12(\openone+\sigma_z),\\
 \rho_{2} &= \frac12(\openone- \tfrac{\sqrt{3}}{2} \sigma_x-\tfrac{1}{2}\sigma_z),\\
 \rho_{3} &= \frac12(\openone+\tfrac{\sqrt{3}}{2}\sigma_x-\tfrac{1}{2}\sigma_z),
\end{split}\end{equation}
and the measurement effects
$M_{1|1} = \frac{2}{3}(\openone-\rho_3)$,
$M_{2|1} = \frac{2}{3}(\openone-\rho_2)$,
$M_{3|1} = \frac{2}{3}(\openone-\rho_1)$,
$M_{1|2} = \openone-\rho_1$, $M_{2|2} =\rho_1$.

\bibliography{bib}

\begin{thebibliography}{23}%
\makeatletter
\providecommand \@ifxundefined [1]{%
 \@ifx{#1\undefined}
}%
\providecommand \@ifnum [1]{%
 \ifnum #1\expandafter \@firstoftwo
 \else \expandafter \@secondoftwo
 \fi
}%
\providecommand \@ifx [1]{%
 \ifx #1\expandafter \@firstoftwo
 \else \expandafter \@secondoftwo
 \fi
}%
\providecommand \natexlab [1]{#1}%
\providecommand \enquote  [1]{``#1''}%
\providecommand \bibnamefont  [1]{#1}%
\providecommand \bibfnamefont [1]{#1}%
\providecommand \citenamefont [1]{#1}%
\providecommand \href@noop [0]{\@secondoftwo}%
\providecommand \href [0]{\begingroup \@sanitize@url \@href}%
\providecommand \@href[1]{\@@startlink{#1}\@@href}%
\providecommand \@@href[1]{\endgroup#1\@@endlink}%
\providecommand \@sanitize@url [0]{\catcode `\\12\catcode `\$12\catcode
  `\&12\catcode `\#12\catcode `\^12\catcode `\_12\catcode `\%12\relax}%
\providecommand \@@startlink[1]{}%
\providecommand \@@endlink[0]{}%
\providecommand \url  [0]{\begingroup\@sanitize@url \@url }%
\providecommand \@url [1]{\endgroup\@href {#1}{\urlprefix }}%
\providecommand \urlprefix  [0]{URL }%
\providecommand \Eprint [0]{\href }%
\providecommand \doibase [0]{http://dx.doi.org/}%
\providecommand \selectlanguage [0]{\@gobble}%
\providecommand \bibinfo  [0]{\@secondoftwo}%
\providecommand \bibfield  [0]{\@secondoftwo}%
\providecommand \translation [1]{[#1]}%
\providecommand \BibitemOpen [0]{}%
\providecommand \bibitemStop [0]{}%
\providecommand \bibitemNoStop [0]{.\EOS\space}%
\providecommand \EOS [0]{\spacefactor3000\relax}%
\providecommand \BibitemShut  [1]{\csname bibitem#1\endcsname}%
\let\auto@bib@innerbib\@empty
\bibitem [{\citenamefont {Bacon}\ \emph {et~al.}(2006)\citenamefont {Bacon},
  \citenamefont {Childs},\ and\ \citenamefont {van Dam}}]{A21}%
  \BibitemOpen
  \bibfield  {author} {\bibinfo {author} {\bibfnamefont {D.}~\bibnamefont
  {Bacon}}, \bibinfo {author} {\bibfnamefont {A.~M.}\ \bibnamefont {Childs}}, \
  and\ \bibinfo {author} {\bibfnamefont {W.}~\bibnamefont {van Dam}},\
  }\bibfield  {title} {\enquote {\bibinfo {title} {Optimal measurements for the
  dihedral hidden subgroup problem},}\ }\href {\doibase 10.4086/cjtcs.2006.002}
  {\bibfield  {journal} {\bibinfo  {journal} {Chicago J. Theo. Comput. Sci.}\
  }\textbf {\bibinfo {volume} {2006}},\ \bibinfo {pages} {2} (\bibinfo {year}
  {2006})}\BibitemShut {NoStop}%
\bibitem [{\citenamefont {Bacon}\ \emph {et~al.}(2005)\citenamefont {Bacon},
  \citenamefont {Childs},\ and\ \citenamefont {van Dam}}]{A22}%
  \BibitemOpen
  \bibfield  {author} {\bibinfo {author} {\bibfnamefont {D.}~\bibnamefont
  {Bacon}}, \bibinfo {author} {\bibfnamefont {A.~M.}\ \bibnamefont {Childs}}, \
  and\ \bibinfo {author} {\bibfnamefont {W.}~\bibnamefont {van Dam}},\
  }\bibfield  {title} {\enquote {\bibinfo {title} {From optimal measurement to
  efficient quantum algorithms for the hidden subgroup problem over semidirect
  product groups},}\ }in\ \href {\doibase 10.1109/SFCS.2005.38} {\emph
  {\bibinfo {booktitle} {46th Annual IEEE Symposium on Foundations of Computer
  Science (FOCS'05)}}}\ (\bibinfo {year} {2005})\ pp.\ \bibinfo {pages}
  {469--478}\BibitemShut {NoStop}%
\bibitem [{\citenamefont {Renes}(2004)}]{A25}%
  \BibitemOpen
  \bibfield  {author} {\bibinfo {author} {\bibfnamefont {J.~M.}\ \bibnamefont
  {Renes}},\ }\bibfield  {title} {\enquote {\bibinfo {title} {Spherical-code
  key-distribution protocols for qubits},}\ }\href {\doibase
  10.1103/PhysRevA.70.052314} {\bibfield  {journal} {\bibinfo  {journal} {Phys.
  Rev. A}\ }\textbf {\bibinfo {volume} {70}},\ \bibinfo {pages} {052314}
  (\bibinfo {year} {2004})}\BibitemShut {NoStop}%
\bibitem [{\citenamefont {Bennett}(1992)}]{A26}%
  \BibitemOpen
  \bibfield  {author} {\bibinfo {author} {\bibfnamefont {C.~H.}\ \bibnamefont
  {Bennett}},\ }\bibfield  {title} {\enquote {\bibinfo {title} {Quantum
  cryptography using any two nonorthogonal states},}\ }\href {\doibase
  10.1103/PhysRevLett.68.3121} {\bibfield  {journal} {\bibinfo  {journal}
  {Phys. Rev. Lett.}\ }\textbf {\bibinfo {volume} {68}},\ \bibinfo {pages}
  {3121} (\bibinfo {year} {1992})}\BibitemShut {NoStop}%
\bibitem [{\citenamefont {Ac{\'\i}n}\ \emph {et~al.}(2016)\citenamefont
  {Ac{\'\i}n}, \citenamefont {Pironio}, \citenamefont {V{\'e}rtesi},\ and\
  \citenamefont {Wittek}}]{A24}%
  \BibitemOpen
  \bibfield  {author} {\bibinfo {author} {\bibfnamefont {A.}~\bibnamefont
  {Ac{\'\i}n}}, \bibinfo {author} {\bibfnamefont {S.}~\bibnamefont {Pironio}},
  \bibinfo {author} {\bibfnamefont {T.}~\bibnamefont {V{\'e}rtesi}}, \ and\
  \bibinfo {author} {\bibfnamefont {P.}~\bibnamefont {Wittek}},\ }\bibfield
  {title} {\enquote {\bibinfo {title} {Optimal randomness certification from
  one entangled bit},}\ }\href {\doibase 10.1103/PhysRevA.93.040102} {\bibfield
   {journal} {\bibinfo  {journal} {Phys. Rev. A}\ }\textbf {\bibinfo {volume}
  {93}},\ \bibinfo {pages} {040102} (\bibinfo {year} {2016})}\BibitemShut
  {NoStop}%
\bibitem [{\citenamefont {Bisio}\ \emph {et~al.}(2009)\citenamefont {Bisio},
  \citenamefont {Chiribella}, \citenamefont {D'Ariano}, \citenamefont
  {Facchini},\ and\ \citenamefont {Perinotti}}]{A23}%
  \BibitemOpen
  \bibfield  {author} {\bibinfo {author} {\bibfnamefont {A.}~\bibnamefont
  {Bisio}}, \bibinfo {author} {\bibfnamefont {G.}~\bibnamefont {Chiribella}},
  \bibinfo {author} {\bibfnamefont {G.~M.}\ \bibnamefont {D'Ariano}}, \bibinfo
  {author} {\bibfnamefont {S.}~\bibnamefont {Facchini}}, \ and\ \bibinfo
  {author} {\bibfnamefont {P.}~\bibnamefont {Perinotti}},\ }\bibfield  {title}
  {\enquote {\bibinfo {title} {Optimal quantum tomography},}\ }\href {\doibase
  10.1109/JSTQE.2009.2029243} {\bibfield  {journal} {\bibinfo  {journal} {IEEE
  J. Sel. Top. Quantum Electron.}\ }\textbf {\bibinfo {volume} {15}},\ \bibinfo
  {pages} {1646--1660} (\bibinfo {year} {2009})}\BibitemShut {NoStop}%
\bibitem [{\citenamefont {Peres}(1993)}]{Peres:1993}%
  \BibitemOpen
  \bibfield  {author} {\bibinfo {author} {\bibfnamefont {A.}~\bibnamefont
  {Peres}},\ }\href@noop {} {\emph {\bibinfo {title} {Quantum Theory: Concepts
  and Methods}}}\ (\bibinfo  {publisher} {Kluwer},\ \bibinfo {address}
  {Dordrecht},\ \bibinfo {year} {1993})\BibitemShut {NoStop}%
\bibitem [{\citenamefont {Mironowicz}\ and\ \citenamefont
  {Paw\l{}owski}(2019)}]{Mironowicz2019SW}%
  \BibitemOpen
  \bibfield  {author} {\bibinfo {author} {\bibfnamefont {Piotr}\ \bibnamefont
  {Mironowicz}}\ and\ \bibinfo {author} {\bibfnamefont {Marcin}\ \bibnamefont
  {Paw\l{}owski}},\ }\bibfield  {title} {\enquote {\bibinfo {title}
  {Experimentally feasible semi-device-independent certification of
  four-outcome positive-operator-valued measurements},}\ }\href {\doibase
  10.1103/PhysRevA.100.030301} {\bibfield  {journal} {\bibinfo  {journal}
  {Phys. Rev. A}\ }\textbf {\bibinfo {volume} {100}},\ \bibinfo {pages}
  {030301} (\bibinfo {year} {2019})}\BibitemShut {NoStop}%
\bibitem [{\citenamefont {Tavakoli}(2020)}]{Tavakoli2020bSW}%
  \BibitemOpen
  \bibfield  {author} {\bibinfo {author} {\bibfnamefont {Armin}\ \bibnamefont
  {Tavakoli}},\ }\bibfield  {title} {\enquote {\bibinfo {title}
  {Semi-device-independent certification of independent quantum state and
  measurement devices},}\ }\href {\doibase 10.1103/PhysRevLett.125.150503}
  {\bibfield  {journal} {\bibinfo  {journal} {Phys. Rev. Lett.}\ }\textbf
  {\bibinfo {volume} {125}},\ \bibinfo {pages} {150503} (\bibinfo {year}
  {2020})}\BibitemShut {NoStop}%
\bibitem [{\citenamefont {Miklin}\ \emph {et~al.}(2020)\citenamefont {Miklin},
  \citenamefont {Borka\l{}a},\ and\ \citenamefont
  {Paw\l{}owski}}]{Miklin2020SW}%
  \BibitemOpen
  \bibfield  {author} {\bibinfo {author} {\bibfnamefont {Nikolai}\ \bibnamefont
  {Miklin}}, \bibinfo {author} {\bibfnamefont {Jakub~J.}\ \bibnamefont
  {Borka\l{}a}}, \ and\ \bibinfo {author} {\bibfnamefont {Marcin}\ \bibnamefont
  {Paw\l{}owski}},\ }\bibfield  {title} {\enquote {\bibinfo {title}
  {Semi-device-independent self-testing of unsharp measurements},}\ }\href
  {\doibase 10.1103/PhysRevResearch.2.033014} {\bibfield  {journal} {\bibinfo
  {journal} {Phys. Rev. Research}\ }\textbf {\bibinfo {volume} {2}},\ \bibinfo
  {pages} {033014} (\bibinfo {year} {2020})}\BibitemShut {NoStop}%
\bibitem [{\citenamefont {G{\'o}mez}\ \emph {et~al.}(2016)\citenamefont
  {G{\'o}mez}, \citenamefont {G{\'o}mez}, \citenamefont {Gonz{\'a}lez},
  \citenamefont {Ca{\~n}as}, \citenamefont {Barra}, \citenamefont {Delgado},
  \citenamefont {Xavier}, \citenamefont {Cabello}, \citenamefont {Kleinmann},
  \citenamefont {V{\'e}rtesi},\ and\ \citenamefont {Lima}}]{A12}%
  \BibitemOpen
  \bibfield  {author} {\bibinfo {author} {\bibfnamefont {Esteban~S.}\
  \bibnamefont {G{\'o}mez}}, \bibinfo {author} {\bibfnamefont {Santiago}\
  \bibnamefont {G{\'o}mez}}, \bibinfo {author} {\bibfnamefont {Pablo}\
  \bibnamefont {Gonz{\'a}lez}}, \bibinfo {author} {\bibfnamefont {Gustavo}\
  \bibnamefont {Ca{\~n}as}}, \bibinfo {author} {\bibfnamefont {Johanna~F.}\
  \bibnamefont {Barra}}, \bibinfo {author} {\bibfnamefont {Aldo}\ \bibnamefont
  {Delgado}}, \bibinfo {author} {\bibfnamefont {Guilherme~B.}\ \bibnamefont
  {Xavier}}, \bibinfo {author} {\bibfnamefont {Ad{\'a}n}\ \bibnamefont
  {Cabello}}, \bibinfo {author} {\bibfnamefont {Matthias}\ \bibnamefont
  {Kleinmann}}, \bibinfo {author} {\bibfnamefont {Tam{\'a}s}\ \bibnamefont
  {V{\'e}rtesi}}, \ and\ \bibinfo {author} {\bibfnamefont {Gustavo}\
  \bibnamefont {Lima}},\ }\bibfield  {title} {\enquote {\bibinfo {title}
  {Device-independent certification of a nonprojective qubit measurement},}\
  }\href {\doibase 10.1103/PhysRevLett.117.260401} {\bibfield  {journal}
  {\bibinfo  {journal} {Phys. Rev. Lett.}\ }\textbf {\bibinfo {volume} {117}},\
  \bibinfo {pages} {260401} (\bibinfo {year} {2016})}\BibitemShut {NoStop}%
\bibitem [{\citenamefont {Tavakoli}\ \emph {et~al.}(2020)\citenamefont
  {Tavakoli}, \citenamefont {Smania}, \citenamefont {V{\'e}rtesi},
  \citenamefont {Brunner},\ and\ \citenamefont {Bourennane}}]{Tavakoli2020aSW}%
  \BibitemOpen
  \bibfield  {author} {\bibinfo {author} {\bibfnamefont {Armin}\ \bibnamefont
  {Tavakoli}}, \bibinfo {author} {\bibfnamefont {Massimiliano}\ \bibnamefont
  {Smania}}, \bibinfo {author} {\bibfnamefont {Tam{\'a}s}\ \bibnamefont
  {V{\'e}rtesi}}, \bibinfo {author} {\bibfnamefont {Nicolas}\ \bibnamefont
  {Brunner}}, \ and\ \bibinfo {author} {\bibfnamefont {Mohamed}\ \bibnamefont
  {Bourennane}},\ }\bibfield  {title} {\enquote {\bibinfo {title} {Self-testing
  nonprojective quantum measurements in prepare-and-measure experiments},}\
  }\href {\doibase 10.1126/sciadv.aaw6664} {\bibfield  {journal} {\bibinfo
  {journal} {Sci. Adv.}\ }\textbf {\bibinfo {volume} {6}},\ \bibinfo {pages}
  {eaaw6664} (\bibinfo {year} {2020})}\BibitemShut {NoStop}%
\bibitem [{\citenamefont {Smania}\ \emph {et~al.}(2020)\citenamefont {Smania},
  \citenamefont {Mironowicz}, \citenamefont {Nawareg}, \citenamefont
  {Paw\l{}owski}, \citenamefont {Cabello},\ and\ \citenamefont
  {Bourennane}}]{Smania2020}%
  \BibitemOpen
  \bibfield  {author} {\bibinfo {author} {\bibfnamefont {Massimiliano}\
  \bibnamefont {Smania}}, \bibinfo {author} {\bibfnamefont {Piotr}\
  \bibnamefont {Mironowicz}}, \bibinfo {author} {\bibfnamefont {Mohamed}\
  \bibnamefont {Nawareg}}, \bibinfo {author} {\bibfnamefont {Marcin}\
  \bibnamefont {Paw\l{}owski}}, \bibinfo {author} {\bibfnamefont {Ad{\'a}n}\
  \bibnamefont {Cabello}}, \ and\ \bibinfo {author} {\bibfnamefont {Mohamed}\
  \bibnamefont {Bourennane}},\ }\bibfield  {title} {\enquote {\bibinfo {title}
  {Experimental certification of an informationally complete quantum
  measurement in a device-independent protocol},}\ }\href {\doibase
  10.1364/OPTICA.377959} {\bibfield  {journal} {\bibinfo  {journal} {Optica}\
  }\textbf {\bibinfo {volume} {7}},\ \bibinfo {pages} {123--128} (\bibinfo
  {year} {2020})}\BibitemShut {NoStop}%
\bibitem [{\citenamefont {Sentis}\ \emph {et~al.}(2013)\citenamefont {Sentis},
  \citenamefont {Gendra}, \citenamefont {Bartlett},\ and\ \citenamefont
  {Doherty}}]{A31}%
  \BibitemOpen
  \bibfield  {author} {\bibinfo {author} {\bibfnamefont {G.}~\bibnamefont
  {Sentis}}, \bibinfo {author} {\bibfnamefont {B.}~\bibnamefont {Gendra}},
  \bibinfo {author} {\bibfnamefont {S.~D.}\ \bibnamefont {Bartlett}}, \ and\
  \bibinfo {author} {\bibfnamefont {A.~C.}\ \bibnamefont {Doherty}},\
  }\bibfield  {title} {\enquote {\bibinfo {title} {Decomposition of any quantum
  measurement into extremals},}\ }\href {\doibase
  10.1088/1751-8113/46/37/375302} {\bibfield  {journal} {\bibinfo  {journal}
  {J. Phys. A}\ }\textbf {\bibinfo {volume} {46}},\ \bibinfo {pages} {375302}
  (\bibinfo {year} {2013})}\BibitemShut {NoStop}%
\bibitem [{\citenamefont {Oszmaniec}\ \emph {et~al.}(2017)\citenamefont
  {Oszmaniec}, \citenamefont {Guerini}, \citenamefont {Wittek},\ and\
  \citenamefont {Ac{\'\i}n}}]{A15}%
  \BibitemOpen
  \bibfield  {author} {\bibinfo {author} {\bibfnamefont {M.}~\bibnamefont
  {Oszmaniec}}, \bibinfo {author} {\bibfnamefont {L.}~\bibnamefont {Guerini}},
  \bibinfo {author} {\bibfnamefont {P.}~\bibnamefont {Wittek}}, \ and\ \bibinfo
  {author} {\bibfnamefont {A.}~\bibnamefont {Ac{\'\i}n}},\ }\bibfield  {title}
  {\enquote {\bibinfo {title} {Simulating positive-operator-valued measures
  with projective measurements},}\ }\href {\doibase
  10.1103/PhysRevLett.119.190501} {\bibfield  {journal} {\bibinfo  {journal}
  {Phys. Rev. Lett.}\ }\textbf {\bibinfo {volume} {119}},\ \bibinfo {pages}
  {190501} (\bibinfo {year} {2017})}\BibitemShut {NoStop}%
\bibitem [{\citenamefont {Busch}\ \emph {et~al.}(2009)\citenamefont {Busch},
  \citenamefont {Grabowski},\ and\ \citenamefont {Lahti}}]{Busch1995}%
  \BibitemOpen
  \bibfield  {author} {\bibinfo {author} {\bibfnamefont {P.}~\bibnamefont
  {Busch}}, \bibinfo {author} {\bibfnamefont {M.}~\bibnamefont {Grabowski}}, \
  and\ \bibinfo {author} {\bibfnamefont {P.~J.}\ \bibnamefont {Lahti}},\
  }\enquote {\bibinfo {title} {The spectral theorem},}\ in\ \href {\doibase
  10.1090/gsm/099} {\emph {\bibinfo {booktitle} {Operational Quantum
  Physics}}}\ (\bibinfo  {publisher} {American Mathematical Society},\ \bibinfo
  {address} {Berlin},\ \bibinfo {year} {2009})\ Chap.~\bibinfo {chapter} {3},
  p.~\bibinfo {pages} {99}\BibitemShut {NoStop}%
\bibitem [{\citenamefont {Ivanovic}(1987)}]{A3}%
  \BibitemOpen
  \bibfield  {author} {\bibinfo {author} {\bibfnamefont {I.~D.}\ \bibnamefont
  {Ivanovic}},\ }\bibfield  {title} {\enquote {\bibinfo {title} {How to
  differentiate between non-orthogonal states},}\ }\href {\doibase
  10.1016/0375-9601(87)90222-2} {\bibfield  {journal} {\bibinfo  {journal}
  {Phys. Lett. A}\ }\textbf {\bibinfo {volume} {123}},\ \bibinfo {pages} {257}
  (\bibinfo {year} {1987})}\BibitemShut {NoStop}%
\bibitem [{\citenamefont {Hoffmann}\ \emph {et~al.}(2018)\citenamefont
  {Hoffmann}, \citenamefont {Spee}, \citenamefont {G{\"{u}}hne},\ and\
  \citenamefont {Budroni}}]{Hoffmann:2018NJP}%
  \BibitemOpen
  \bibfield  {author} {\bibinfo {author} {\bibfnamefont {Jannik}\ \bibnamefont
  {Hoffmann}}, \bibinfo {author} {\bibfnamefont {Cornelia}\ \bibnamefont
  {Spee}}, \bibinfo {author} {\bibfnamefont {Otfried}\ \bibnamefont
  {G{\"{u}}hne}}, \ and\ \bibinfo {author} {\bibfnamefont {Costantino}\
  \bibnamefont {Budroni}},\ }\bibfield  {title} {\enquote {\bibinfo {title}
  {Structure of temporal correlations of a qubit},}\ }\href {\doibase
  10.1088/1367-2630/aae87f} {\bibfield  {journal} {\bibinfo  {journal} {New J.
  Phys.}\ }\textbf {\bibinfo {volume} {20}},\ \bibinfo {pages} {102001}
  (\bibinfo {year} {2018})}\BibitemShut {NoStop}%
\bibitem [{\citenamefont {Frenkel}\ and\ \citenamefont {Weiner}(2015)}]{A7}%
  \BibitemOpen
  \bibfield  {author} {\bibinfo {author} {\bibfnamefont {P.~E.}\ \bibnamefont
  {Frenkel}}\ and\ \bibinfo {author} {\bibfnamefont {M.}~\bibnamefont
  {Weiner}},\ }\bibfield  {title} {\enquote {\bibinfo {title} {Classical
  information storage in an $n$-level quantum system},}\ }\href {\doibase
  10.1007/s00220-015-2463-0} {\bibfield  {journal} {\bibinfo  {journal} {Comm.
  Math. Phys.}\ }\textbf {\bibinfo {volume} {340}},\ \bibinfo {pages}
  {563--574} (\bibinfo {year} {2015})}\BibitemShut {NoStop}%
\bibitem [{\citenamefont {Kleinmann}\ and\ \citenamefont {Cabello}(2016)}]{A6}%
  \BibitemOpen
  \bibfield  {author} {\bibinfo {author} {\bibfnamefont {M.}~\bibnamefont
  {Kleinmann}}\ and\ \bibinfo {author} {\bibfnamefont {A.}~\bibnamefont
  {Cabello}},\ }\bibfield  {title} {\enquote {\bibinfo {title} {Quantum
  correlations are stronger than all nonsignaling correlarions produced by
  $n$-outcome measurements},}\ }\href {\doibase 10.1103/PhysRevLett.117.150401}
  {\bibfield  {journal} {\bibinfo  {journal} {Phys. Rev. Lett.}\ }\textbf
  {\bibinfo {volume} {117}},\ \bibinfo {pages} {150401} (\bibinfo {year}
  {2016})}\BibitemShut {NoStop}%
\bibitem [{\citenamefont {P\'al}\ and\ \citenamefont
  {V\'ertesi}(2010)}]{Pal2010}%
  \BibitemOpen
  \bibfield  {author} {\bibinfo {author} {\bibfnamefont {K\'aroly~F.}\
  \bibnamefont {P\'al}}\ and\ \bibinfo {author} {\bibfnamefont {Tam\'as}\
  \bibnamefont {V\'ertesi}},\ }\bibfield  {title} {\enquote {\bibinfo {title}
  {Maximal violation of a bipartite three-setting, two-outcome {B}ell
  inequality using infinite-dimensional quantum systems},}\ }\href {\doibase
  10.1103/PhysRevA.82.022116} {\bibfield  {journal} {\bibinfo  {journal} {Phys.
  Rev. A}\ }\textbf {\bibinfo {volume} {82}},\ \bibinfo {pages} {022116}
  (\bibinfo {year} {2010})}\BibitemShut {NoStop}%
\bibitem [{\citenamefont {Brunner}\ \emph {et~al.}(2013)\citenamefont
  {Brunner}, \citenamefont {Navasgues},\ and\ \citenamefont
  {V{\'e}rtesi}}]{A2}%
  \BibitemOpen
  \bibfield  {author} {\bibinfo {author} {\bibfnamefont {N.}~\bibnamefont
  {Brunner}}, \bibinfo {author} {\bibfnamefont {M.}~\bibnamefont {Navasgues}},
  \ and\ \bibinfo {author} {\bibfnamefont {T.}~\bibnamefont {V{\'e}rtesi}},\
  }\bibfield  {title} {\enquote {\bibinfo {title} {Dimension witnesses and
  quantum state discrimination},}\ }\href {\doibase
  10.1103/PhysRevLett.110.150501} {\bibfield  {journal} {\bibinfo  {journal}
  {Phys. Rev. Lett.}\ }\textbf {\bibinfo {volume} {110}},\ \bibinfo {pages}
  {150501} (\bibinfo {year} {2013})}\BibitemShut {NoStop}%
\bibitem [{\citenamefont {Mazurek}\ \emph {et~al.}(2021)\citenamefont
  {Mazurek}, \citenamefont {Pusey}, \citenamefont {Resch},\ and\ \citenamefont
  {Spekkens}}]{A1}%
  \BibitemOpen
  \bibfield  {author} {\bibinfo {author} {\bibfnamefont {Michael~D.}\
  \bibnamefont {Mazurek}}, \bibinfo {author} {\bibfnamefont {Matthew~F.}\
  \bibnamefont {Pusey}}, \bibinfo {author} {\bibfnamefont {Kevin~J.}\
  \bibnamefont {Resch}}, \ and\ \bibinfo {author} {\bibfnamefont {Robert~W.}\
  \bibnamefont {Spekkens}},\ }\bibfield  {title} {\enquote {\bibinfo {title}
  {Experimentally bounding deviations from quantum theory in the landscape of
  generalized probabilistic theories},}\ }\href {\doibase
  10.1103/PRXQuantum.2.020302} {\bibfield  {journal} {\bibinfo  {journal} {PRX
  Quantum}\ }\textbf {\bibinfo {volume} {2}},\ \bibinfo {pages} {020302}
  (\bibinfo {year} {2021})}\BibitemShut {NoStop}%
\end{thebibliography}%

\end{document}